\documentclass[conference,letterpaper]{IEEEtran}

\usepackage[utf8]{inputenc} 
\usepackage[T1]{fontenc}
\usepackage{url}
\usepackage{ifthen}
\usepackage{cite}
\usepackage{braket}
\usepackage{amssymb}
\usepackage{amsthm}
\usepackage{graphicx}
\usepackage{dsfont}
\usepackage{comment}
\usepackage{tikz}
\usepackage[cmex10]{amsmath} 
\usepackage{hyperref}
\newtheorem{theorem}{Theorem}
\newtheorem{definition}{Definition}

\newtheorem{lemma}{Lemma}

\newtheorem{corollary}{Corollary}

\def\a{\mathbf{a}}
\def\b{\mathbf{b}}
\def\u{\mathbf{u}}

\def\U{\mathbf{U}}

\def\F{\mathbb{F}}

\def\cA{\mathcal{A}}
\def\cC{\mathcal{C}}

\def\rmG{\mathrm{G}}

\begin{document}
\title{Lifting an $[n,k,d]$ Binary Linear Code to an $[[n,k,d]]$ Quantum Stabilizer Code}


\author{%
  \IEEEauthorblockN{Thamadathil Aswanth and Navin Kashyap}
  \IEEEauthorblockA{Dept. Electrical Communication Engineering, Indian Institute of Science\\
                    \{aswantht, nkashyap\}@iisc.ac.in}
}


\maketitle

\begin{abstract}
    The codeword stabilized (CWS) framework provides a systematic method for constructing quantum error correcting codes from a simple undirected graph and a classical code. When the underlying classical code is linear, the resulting quantum code is a stabilizer code. In this work, we investigate the conditions under which an arbitrary $[n,k,d]$ classical binary linear code can be lifted to an $[[n,k,d]]$ quantum stabilizer code using the CWS framework. We analyze the random CWS code constructed from a given $[n,k,d]$ binary linear code $\mathcal{C}$ and an Erd\H{o}s-R\'enyi random graph $\mathrm{G}(n, \frac12)$. We establish a lower bound on the probability that the resulting quantum code preserves the minimum distance $d$ of the classical code. This probabilistic analysis yields a Gilbert-Varshamov-type condition on the $n,k,d$ parameters of the classical code that guarantees the existence of a graph resulting in a quantum code with the same $n,k,d$ parameters. 
\end{abstract}

\section{Introduction}
Quantum computing provides exponential speedups for certain computational problems by leveraging quantum mechanical phenomena such as superposition and entanglement \cite{Nielsen_Chuang_2010, Shor1994, Shor_1997}. However, quantum states are fragile and susceptible to decoherence, and quantum operations are often faulty. Quantum error correcting codes (QECC) are necessary to address these problems and realize reliable fault-tolerant quantum computation by encoding quantum information into a larger entangled subspace \cite{PhysRevA.52.R2493, Knill97, knill1999theory}.

Quantum stabilizer codes are the most well-studied class of QECCs \cite{gottesman1997stabilizer}. By defining a quantum code as the simultaneous $+1$ eigenspace of an abelian subgroup of the Pauli group, the stabilizer formalism presents a discrete algebraic framework for quantum error correction. The Calderbank-Shor-Steane code, which falls in the class of quantum stabilizer codes, provided a method to construct quantum codes from classical codes \cite{CalderbankShorCSS, SteaneCSS}. This construction uses two classical codes, with the dual of one code contained in the other.

Another approach to derive QECCs from classical codes is via the codeword stabilized (CWS) codes framework \cite{CWSCode, CWSCode2, CWS_GR-chapter-in-LB}. Inputs to a CWS construction are two mathematical objects: a graph $\mathrm G$, which defines a graph  \cite{CWS_PhysRevA.65.012308, oneWayQC}, and a classical code, not necessarily linear, which specifies a set of Pauli $Z$ operators that define a basis for the quantum code. When the classical code used is strictly non-linear, the quantum code obtained is not a stabilizer code, and hence CWS codes are a generalization of stabilizer codes \cite{CWS_Grassl1, CWS_grassl2}. When the classical code is linear, the CWS codes are stabilizer codes. For the purpose of this work, we will be considering CWS codes constructed from classical binary linear codes.

Kapshikar and Kundu used the CWS framework to show that computing the minimum distance of a quantum stabilizer code, exactly or approximately, is NP-hard \cite{kapshikarKundu}. Given an $[n,k,d]$ binary linear code $\mathcal{C}$, they provide an explicit construction of a graph $\mathrm G$ such that $\mathcal{C}$ and $\mathrm{G}$ give rise to an $[[\mathcal O(n^2), k, d]]$ quantum code. Recently, Kapshikar extended this to show that approximating the quantum code distance to a linear additive gap is NP-hard under randomized reductions, overcoming previous $\mathcal{O}(\sqrt{N})$ barriers \cite{kapshikar2026hardnessapproximatingquantumcode}.

The question we address in this study is the following: given a classical binary linear code $\mathcal{C}$ with parameters $[n,k,d]$, under what conditions can we find a graph $\mathrm{G}$ such that the resulting CWS code $\mathcal{Q} = \operatorname{CWS}(\mathrm{G}, \mathcal{C})$ is an $[[n,k,d]]$ quantum stabilizer code? Constructing such a graph deterministically is highly non-trivial, as the graph must satisfy a set of algebraic constraints to ensure that all low-weight Pauli errors either fall outside the classical code or commute with a set of Pauli $Z$ operators, derived from $\mathcal{C}$.

In our work, given an $[n,k,d]$ binary linear code $\mathcal{C}$, we consider the random CWS code obtained from $\mathcal{C}$ and an Erd\H{o}s-R\'{e}nyi model random graph $\mathrm G$ \cite{RandomGraphs_erdos, RandomGraph_Gilbert}. We give a lower bound on the probability that the resulting CWS code has minimum distance $d$. From this, we obtain a Gilbert-Varshamov (GV) style bound on the parameters of the classical code $\cC$ that guarantees the existence of a graph $\rmG$ that yields a quantum code with the same $n,k,d$ parameters as $\cC$. Furthermore, we show that for a family of classical codes whose asymptotic rate $R$ and relative distance $\delta$ satisfy $R + h_2(\delta) + \delta \log_2(3) < 1$, the random CWS codes obtained from Erdős-Rényi random graphs have rate $R$ and relative distance $\delta$ with a probability approaching 1 as $n \to \infty$.

We establish the notations used and the mathematical preliminaries in Section~\ref{Notations}. This is followed by the main contributions of this work in Section~\ref{sec:main}.

\section{Notation and Preliminaries} \label{Notations}
\subsection{Notation}
Let $\mathbb{F}_2$ denote the finite field of order 2, and let $\mathbb{C}$ denote the field of complex numbers. The $m \times m$ identity matrix over $\mathbb{F}_2$ is denoted by $\mathbf{I}_m$. For binary vectors $\mathbf{a} = (a_1,a_2,\ldots,a_n)$ and $\mathbf{b} = (b_1,b_2,\ldots,b_n)$, their dot product over $\F_2$ is defined as $\mathbf{a} \cdot \mathbf{b} = \sum_{i=1}^n a_ib_i \mod 2$. For any positive integer $n$, we define $[n] := \{1, 2, \ldots, n\}$. We use $\mathbb{P}(\cdot)$ to denote the probability of an event.

\subsection{Binary Linear Codes}
A classical binary linear code $\mathcal{C}$ is a $k$-dimensional subspace of $\mathbb{F}_2^n$, where $k$ is referred to as the dimension and $n$ as the blocklength of the code. The elements of $\mathcal{C}$ are called codewords, and the minimum distance $d$ of $\mathcal{C}$ is defined as the minimum Hamming weight of a non-zero codeword in $\mathcal{C}$. A code with minimum distance $d$ can detect up to $d-1$ errors and correct up to $\lfloor \frac{d-1}{2} \rfloor$ errors. A binary linear code of block length $n$, dimension $k$, and minimum distance $d$ is denoted as an $[n, k, d]$ code.

An $[n, k, d]$ code $\mathcal{C}$ can be specified by a $k \times n$ generator matrix $\mathbf{G}$, whose rows form a basis for $\mathcal{C}$. Alternatively, $\mathcal{C}$ can be expressed as the null space of an $(n-k) \times n$ parity-check matrix $\mathbf{H}$. Up to a permutation of the coordinates, a generator matrix can be written in the systematic form $\mathbf{G} = [\mathbf{I}_k \mid \mathbf{P}]$, in which case the corresponding parity-check matrix is given by $\mathbf{H} = [\mathbf{P}^T \mid \mathbf{I}_{n-k}]$, where $\mathbf{P} \in \mathbb{F}_2^{k \times (n-k)}$.

\subsection{Quantum Codes}
An $n$-qubit quantum code $\mathcal{Q}$ is a $K$-dimensional subspace of the Hilbert space $\mathbb{C}^{2^n}$. The codes of interest in this paper fall within the category of quantum stabilizer codes \cite{gottesman1997stabilizer}, whose dimension takes the form $K=2^k$, where $k$ is referred to as the logical dimension of $\mathcal{Q}$. Vectors $\ket{\psi} \in \mathcal{Q}$ are referred to as code states.

The stabilizer formalism is built upon Pauli operators. The single-qubit Pauli matrices are defined as
\begin{equation*}
    X = \begin{bmatrix}
     0 & 1 \\
     1 & 0
    \end{bmatrix}, Y = \begin{bmatrix}
        0 & -i \\
        i & 0
    \end{bmatrix}, Z = \begin{bmatrix}
        1 & 0 \\
        0 & -1
    \end{bmatrix}, \text{and} I = \begin{bmatrix}
        1 & 0 \\
        0 & 1
    \end{bmatrix}.
\end{equation*}
Each of these operators squares to the identity operator $I$. Furthermore, they satisfy the relation $Y=iXZ$, and the $X$ and $Z$ operators anti-commute: $XZ=-ZX$. For an $n$-qubit system with qubits indexed from $1$ to $n$, a single-qubit Pauli operator $\sigma \in \{I, X, Y, Z\}$ acting on the $j$-th qubit is denoted by $\sigma_j$.

One can take the $n$-fold tensor product of single-qubit Pauli operators and construct the $n$-qubit Pauli group as $\mathcal{P}_n := \{i^\ell \sigma_1 \otimes \sigma_2\otimes\cdots\otimes\sigma_n: \sigma_j \in \{X,Y,Z,I\} \text{ and }\ell \in \{0,1,2,3\}\}$, with multiplication as the group operation.

Up to a global phase factor, any operator $M = i^{\ell} \sigma_1 \otimes \sigma_2 \otimes \cdots \otimes \sigma_n \in \mathcal{P}_n$ can be mapped to a symplectic binary vector $[\mathbf{a}, \mathbf{b}] \in \mathbb{F}_2^{2n}$, with $\mathbf{a}, \mathbf{b} \in \mathbb{F}_2^n$, where
\begin{equation*}
    (a_j,b_j) = \begin{cases}
        (0,0) & \text{if } \sigma_j = I, \\
        (1,0) & \text{if } \sigma_j = X, \\
        (0,1) & \text{if } \sigma_j = Z, \\
        (1,1) & \text{if } \sigma_j = Y.
    \end{cases}
\end{equation*}
Consequently, we can write $M = i^\ell X(\mathbf{a})Z(\mathbf{b})$, where $X(\mathbf{a}) = X_1^{a_1} X_2^{a_2} \cdots X_n^{a_n}$ and $Z(\mathbf{b}) = Z_1^{b_1} Z_2^{b_2} \cdots Z_n^{b_n}$.

The symplectic weight of such a vector, defined as
\begin{equation*}
    w_{sym}([\mathbf{a}, \mathbf{b}]) := \big| \{j \in [n] \mid (a_j, b_j) \neq (0,0)\} \big|,
\end{equation*}
corresponds to the number of non-identity single-qubit Pauli operators in $M$.

Under the symplectic representation, two Pauli operators $M_1$ and $M_2$ with symplectic binary vectors $[\mathbf{a}^{(1)}, \mathbf{b}^{(1)}]$ and $[\mathbf{a}^{(2)}, \mathbf{b}^{(2)}]$, respectively, commute if and only if their symplectic inner product $\mathbf{a}^{(1)}\cdot \mathbf{b}^{(2)} + \mathbf{b}^{(1)}\cdot \mathbf{a}^{(2)} = 0$. The multiplication of two operators (up to a global phase) is equivalent to the vector addition of their symplectic representations over $\mathbb{F}_2$.

\begin{definition}[Stabilizer Codes]
    A stabilizer group on $n$ qubits is an abelian subgroup $S \subset \mathcal{P}_n$ that does not contain the element $-I$. The stabilizer code defined by $S$ is the subspace of $\mathbb{C}^{2^n}$ that is the simultaneous $+1$ eigenspace of all $s \in S$.
\end{definition}

A stabilizer group is uniquely determined by a set of independent generators. A stabilizer group with $n-k$ independent generators defines a stabilizer code with logical dimension $k$, which is denoted as an $[[n,k]]$ quantum stabilizer code.

In the symplectic representation, an $[[n,k]]$ stabilizer code can be described by $n-k$ linearly independent vectors in $\mathbb{F}_2^{2n}$ corresponding to the independent stabilizer generators. Because the stabilizer generators commute, these vectors have a symplectic inner product of $0$ with each other. The $(n-k)$-dimensional subspace $\mathcal{S} \subset \mathbb{F}_2^{2n}$ spanned by these vectors is equivalent to the stabilizer group. Henceforth, we will refer to stabilizers and other subgroups of the Pauli group using their symplectic representations.

The normalizer of a stabilizer group, expressed in the symplectic representation, is defined as the subspace 
\begin{equation*}
    \mathcal{N}(\mathcal{S}) = \big\{ [\mathbf{u}, \mathbf{v}] \in \mathbb{F}_2^{2n} \mid \mathbf{u} \cdot \mathbf{b} + \mathbf{v} \cdot \mathbf{a} = 0 \quad \forall [\mathbf{a}, \mathbf{b}] \in \mathcal{S} \big\}.
\end{equation*}

For any Pauli operator $N$ whose symplectic representation belongs to $\mathcal{N}(\mathcal{S})$, we have $N\ket{\psi} \in \mathcal{Q}$ for all code states $\ket{\psi} \in \mathcal{Q}$. These operators form the logical Pauli operators of the code $\mathcal{Q}$. Note that $\mathcal{S} \subset \mathcal{N}(\mathcal{S})$, but the action of any element in $\mathcal{S}$ on a code state is trivial. An $[[n,k]]$ quantum stabilizer code with minimum distance $d$ is denoted as an $[[n,k,d]]$ code. 

The minimum distance $d$ of a stabilizer code is defined as the minimum symplectic weight of a vector in $\mathcal{N}(\mathcal{S}) \setminus \mathcal{S}$. Physically, this corresponds to the minimum weight of a Pauli error that maps a code state to another code state, leaving the error undetectable. If no non-trivial stabilizer element in $\mathcal{S}$ has symplectic weight less than the minimum distance $d$, the code is called a pure stabilizer code. 

This work focuses on codeword stabilized (CWS) codes \cite{CWSCode, CWSCode2}, a class of quantum codes that subsumes stabilizer codes. A CWS code on $n$ qubits is defined by a graph $\mathrm{G}$ with $n$ vertices and a classical error correcting code $\mathcal{C}$ of block length $n$. In general, the classical code need not be linear. However, in this paper, we restrict our attention to CWS codes defined by a classical binary linear code $\mathcal{C}$, which are equivalent to stabilizer codes.

Associated with every simple, undirected graph, there exists a unique (up to a global phase) $n$-qubit quantum state known as a graph state.

\begin{definition}[Graph State]
    Consider a simple, undirected graph $\mathrm{G}$ with $n$ vertices and adjacency matrix $\mathbf{B}$. For each vertex $j \in [n]$, define the operator $s_j := X_jZ(\mathbf{b}^j)$, where $\mathbf{b}^j$ is the $j$-th row of $\mathbf{B}$. The graph state $\ket{\mathrm{G}}$ associated with $\mathrm{G}$ is the $n$-qubit state stabilized by the stabilizer group generated by $\{s_j \mid j \in [n]\}$. 
\end{definition}

A CWS code is obtained from a graph state $\ket{\mathrm{G}}$ by applying Pauli $Z$-operators derived from the codewords of a binary linear code $\mathcal{C}$.

\begin{definition}[CWS Code]
    Given a simple, undirected graph $\mathrm{G}$ with $n$ vertices and an $[n,k]$ binary linear code $\mathcal{C}$, the associated CWS code is defined as
    \begin{equation}\label{eq:CWSCodeSpace}
        \text{CWS}(\mathrm{G}, \mathcal{C}) = \operatorname{span} \big\{ Z(\mathbf{c})\ket{\mathrm{G}} \mid \mathbf{c} \in \mathcal{C} \big\}.
    \end{equation}
    Because $\mathcal{C}$ is a linear code, this construction yields an $[[n,k]]$ quantum stabilizer code.
\end{definition}

The following lemma describes the set of errors detectable by a CWS code.

\begin{lemma}[\cite{CWSCode}, Theorem 2]\label{Lem:CWS_Distance}
    Consider a Pauli error $E = X(\mathbf{a})Z(\mathbf{b})$. Let $\mathcal{Q}$ be a CWS code defined by a graph $\mathrm{G}$ with adjacency matrix $\mathbf{B}$ and a binary linear code $\mathcal{C}$. Define\footnote{We follow the $\operatorname{Cl}_{\mathrm{G}}(E)$ notation as introduced in \cite{CWSCode}, where ``Cl'' is supposed to indicate the effective ``classical'' equivalent of the Pauli error $E$, and the subscript indicates its dependence on the underlying graph $\mathrm{G}$.} $\operatorname{Cl}_{\mathrm{G}}(E) := \mathbf{a}\mathbf{B} + \mathbf{b}$. The code $\mathcal{Q}$ detects the error $E$ if and only if the following conditions are satisfied:
    \begin{enumerate}
        \item If $\operatorname{Cl}_{\mathrm{G}}(E) \neq \mathbf{0}$, then $\operatorname{Cl}_{\mathrm{G}}(E) \notin \mathcal{C}$.
        \item If $\operatorname{Cl}_{\mathrm{G}}(E) = \mathbf{0}$, then $E Z(\mathbf{c}) = Z(\mathbf{c}) E \quad \forall \mathbf{c} \in \mathcal{C}$.
    \end{enumerate}
\end{lemma}

\section{Main Results}\label{sec:main}
In this work, we address the following problem: given a classical $[n,k,d]$ binary linear code $\mathcal{C}$, is it possible to construct an $[[n,k,d]]$ quantum stabilizer code within the CWS framework using $\mathcal{C}$? This problem reduces to determining whether there exists an adjacency matrix $\mathbf{B}$ such that the conditions in Lemma~\ref{Lem:CWS_Distance} are satisfied for all Pauli errors $E \in \mathcal{P}_n$ having symplectic weight strictly less than $d$. We approach this by analyzing the probability that a graph $\mathrm{G}$, sampled from an Erd\H{o}s-R\'enyi random graph model \cite{RandomGraphs_erdos, RandomGraph_Gilbert}, satisfies these conditions.

\begin{definition}[Erd\H{o}s-R\'enyi Random Graph]\label{Def:Erdos-Renyi}
    For a positive integer $n$ and $0 \leq p \leq 1$, an Erd\H{o}s-R\'enyi random graph $\mathrm{G}(n,p)$ is a simple, undirected graph on $n$ vertices where each pair of distinct vertices is connected by an edge independently with probability $p$.
\end{definition}

The adjacency matrix of a simple, undirected graph on $n$ vertices is a symmetric $n \times n$ binary matrix with zero diagonal. In other words, it is an $n \times n$ matrix $A = [A_{i,j}]$ with $A_{i,j} = A_{j,i} \in \F_2$ for all $i,j$, and $A_{i,i} = 0$ for all $i$. Define $\cA_n$ to be set of all such matrices. The next lemma records an important property of $\cA_n$ 

\begin{lemma}\label{lem:a_perp}
Let $\a \in \F_2^n$ be a fixed non-zero vector, and define $\a^\perp := \big\{ \b \in \F_2^n \mid \b \cdot \a = 0 \big\}$. Then, $\{\a A: A \in \cA_n\} = \a^\perp$.
\end{lemma}

\begin{proof}
Write $\a = (a_1,a_2,\ldots,a_n)$. For a matrix $A = [A_{i,j}] \in \cA_n$, let $\u = \a A$. We then have
    \begin{equation*}
        \begin{split}
            \mathbf{a} \cdot \mathbf{u} &= \sum_{j = 1}^{n} a_j \left( \sum_{i = 1}^{n} a_i A_{i,j} \right) \\
            &= \sum_{i=1}^n \sum_{j=1}^n a_i a_j A_{i,j} \\
            &= \sum_{i = 1}^{n} a_i^2 A_{i,i} + \sum_{i < j} a_i a_j A_{i,j} + \sum_{i > j} a_i a_j A_{i,j}.
        \end{split}
    \end{equation*}
    Since $A \in \cA_n$, we have $A_{i,i} = 0$ and $A_{i,j} = A_{j,i}$. Consequently, the first term vanishes, and the last two summations are identical and cancel in $\mathbb{F}_2$, giving $\mathbf{a} \cdot \mathbf{u} = 0$. Hence, $\{\a A: A \in \cA_n\} \subseteq \a^\perp$.

    On the other hand, consider any $\b = (b_1,b_2,\ldots,b_n) \in \a^\perp$. Fix an $i$ such that $a_i = 1$; such an $i$ exists since $\a$ is non-zero. Define a matrix $A \in \cA_n$ as follows: for the fixed choice of $i$, set $A_{i,j} = A_{j,i} = b_j$ for all $j \ne i$, and $A_{i,i} = 0$. Further, set $A_{k,\ell} = 0$ for all $k,\ell \ne i$. 
    
    Consider the $\ell$-th entry of $\a A$. When $\ell \ne i$, we have
    $$
    {(\a A)}_{\ell} = \sum_{k} a_k A_{k,\ell} = a_i A_{i,\ell} = b_{\ell},
    $$
    since $A_{k,\ell} = 0$ for all $k \ne i$, $A_{i,\ell} = b_{\ell}$, and $a_i = 1$. When $\ell = i$, we have
    $${(\a A)}_i = \sum_{k \ne i} a_k A_{k,i} = \sum_{k \ne i} a_k b_k = a_i b_i + \a \cdot \b = b_i,$$
    since $\b \in \a^{\perp}$. Thus, $\a A = \b$, from which we conclude that $\{\a A: A \in \cA_n\} = \a^\perp$.
\end{proof}

The lemma below is the key to our analysis of the random CWS code obtained using an Erd\H{o}s-R\'{e}nyi random graph. 

\begin{lemma}\label{lem:uA_distribution}
    Let $\mathbf{a} \in \mathbb{F}_2^n$ be a fixed non-zero vector, and let $\mathbf{A}_n$ be the adjacency matrix of an Erd\H{o}s-R\'enyi random graph $\mathrm{G}(n, \frac12)$. Then, the random vector $\mathbf{U} := \mathbf{a}\mathbf{A}_n$ is uniformly distributed over $\mathbf{a}^\perp := \big\{ \mathbf{b} \in \mathbb{F}_2^n \mid \mathbf{b} \cdot \mathbf{a} = 0 \big\}$.
\end{lemma}

\begin{proof}
    Observe that the adjacency matrix, $\mathbf{A}_n$, of the Erd\H{o}s-R\'enyi random graph $\mathrm{G}(n,\frac12)$ is uniformly distributed over $\cA_n$.
    The mapping $A \mapsto \a A$ defines a linear map $\varphi: \cA_n \to \F_2^n$, the image of which (by Lemma~\ref{lem:a_perp}) is $\a^\perp$. Therefore, the random vector $\mathbf{U} = \varphi(\mathbf{A}_n) = \mathbf{a}\mathbf{A}_n$ is uniformly distributed over the image of the linear map $\varphi$, i.e., $\U \sim \text{Unif}(\a^{\perp})$.
\end{proof}

As the main contribution of this paper, we analyze the probability of obtaining a CWS code that achieves the same minimum distance as its underlying classical code.

\begin{theorem}\label{mainTheorem}
    Let $\mathcal{C}$ be an $[n,k,d]$ classical binary linear code. The probability that $\mathcal{Q} = \operatorname{CWS}\left(\mathcal{C}, \mathrm{G}\left(n, \frac{1}{2}\right)\right)$ is a pure quantum code having minimum distance $d$ is at least 
    \begin{equation*}
        1 - 2^{k-n} \sum_{j=1}^{d-1} \binom{n}{j}(3^j-1).
    \end{equation*}
\end{theorem}

\begin{proof}
    For any Pauli $Z$-type error $E$ with symplectic binary representation $[\mathbf{0},\mathbf{b}]$, we have $\operatorname{Cl}_{\mathrm{G}}(E) = \mathbf{b}$. Because the classical code $\mathcal{C}$ has minimum distance $d$, all such non-trivial errors of symplectic weight less than $d$ satisfy $\mathbf{b} \notin \mathcal{C}$ and are therefore detectable by Lemma~\ref{Lem:CWS_Distance}. Moreover, for any minimum-weight codeword $\mathbf{c}^* \in \mathcal{C}$, the Pauli error $Z(\mathbf{c}^*)$ yields $\operatorname{Cl}_{\mathrm{G}}(Z(\mathbf{c}^*)) = \mathbf{c}^* \in \mathcal{C}$, meaning it is not detectable. Hence, the quantum minimum distance of $\mathcal{Q}$ is at most $d$. Therefore, it suffices to show that all Pauli errors of symplectic weight strictly less than $d$ that contain at least one $X$ or $Y$ component are detectable by $\mathcal{Q}$.
    
    Let $\mathcal{E}$ be the set of all Pauli errors $E = [\mathbf{a}, \mathbf{b}]$ such that the symplectic weight is less than $d$ and $\mathbf{a} \neq \mathbf{0}$. The number of such errors is 
    \begin{equation*}
        |\mathcal{E}| = \sum_{j=1}^{d-1} \binom{n}{j}(3^j-1).
    \end{equation*}
    
    For a fixed $E = [\mathbf{a}, \mathbf{b}] \in \mathcal{E}$, consider the random vector $\mathbf{U} = \mathbf{a}\mathbf{A}_n$, where $\mathbf{A}_n$ is the adjacency matrix of the random graph $\mathrm{G}(n, \frac12)$. From Lemma~\ref{lem:uA_distribution}, we know that $\mathbf{U}$ is uniformly distributed over the subspace $\a^\perp := \{\mathbf{x} \in \mathbb{F}_2^n \mid \mathbf{x} \cdot \mathbf{a} = 0\}$. To ensure that $E$ is detectable, it is sufficient to require that $\operatorname{Cl}_{\mathrm{G}}(E) \notin \mathcal{C}$.

    From Eq.~\eqref{eq:CWSCodeSpace}, it is clear that the stabilizer group of the CWS code is a subgroup of the stabilizer group of the graph state. Therefore, any stabilizer of the code must be of the form $S = [\mathbf{u}, \mathbf{u}\mathbf{A}_n]$ for some $\mathbf{u} \in \mathbb{F}_2^n$, and hence $\operatorname{Cl}_{\mathrm{G}}(S) = \mathbf{0}$; in particular, $\operatorname{Cl}_{\mathrm{G}}(S) \in \mathcal{C}$. Therefore, if $\operatorname{Cl}_{\mathrm{G}}(E) \notin \mathcal{C}$ holds for all $E \in \mathcal{E}$, then no non-trivial stabilizer element of the code has a symplectic weight strictly less than $d$, meaning the code is pure.
    
    Let us define the failure probability for a specific error $E$ as $\mathbb{P}_f(E) := \mathbb{P}(\operatorname{Cl}_{\mathrm{G}}(E) \in \mathcal{C})$. Observe that the event $\{\mathbf{U} + \mathbf{b} \in \mathcal{C}\}$ is equivalent to $\{\mathbf{U} \in (\mathbf{b}+\mathcal{C})\}$, where $\mathbf{b}+\mathcal{C}$ is the coset obtained by translating $\mathcal{C}$ by $\mathbf{b}$. Thus,
    \begin{equation}\label{Eq:failureProbability}
        \mathbb{P}_f(E) = \mathbb{P}\big( \mathbf{U} \in (\mathbf{b}+\mathcal{C}) \big) = \frac{\lvert \a^{\perp} \cap (\mathbf{b}+\mathcal{C}) \rvert}{\lvert \a^{\perp} \rvert}.
    \end{equation}

    By the union bound, we have
    \begin{equation}\label{Eq:SucessProbUnionBound}
    \begin{split}
        \mathbb{P} \left( \bigcap_{E \in \mathcal{E}} \{\operatorname{Cl}_{\mathrm{G}}(E) \notin \mathcal{C}\} \right) &= 1 - \mathbb{P} \left( \bigcup_{E \in \mathcal{E}} \{\operatorname{Cl}_{\mathrm{G}}(E) \in \mathcal{C}\} \right)\\
        &\geq 1 - \sum_{E \in \mathcal{E}} \mathbb{P}_f(E).
        \end{split}
    \end{equation}

    To compute the sum $\sum_{E \in \mathcal{E}} \mathbb{P}_f(E)$, we partition the set $\mathcal{E}$ into disjoint pairs. For any $E = [\mathbf{a}, \mathbf{b}] \in \mathcal{E}$, since $\mathbf{a} \neq \mathbf{0}$, let $i^*$ be the smallest index in $[n]$ such that $a_{i^*} = 1$. We define its paired error $E' = [\mathbf{a}, \mathbf{b}']$ such that $b'_{i^*} = b_{i^*} + 1 \mod 2$, and $b'_i = b_i$ for all $i \neq i^*$. Because $a_{i^*} = 1$, flipping $b_{i^*}$ simply swaps an $X$ operator for a $Y$ operator (or vice versa) on the $i^*$-th qubit. This transformation preserves the symplectic weight, ensuring $E' \in \mathcal{E}$. Applying the same operation to $E'$ gives $E$, and therefore, this pairing partitions $\mathcal{E}$ into exactly $|\mathcal{E}|/2$ disjoint pairs $\{E, E'\}$. 
    
    Now, consider the sum of their failure probabilities, $\mathbb{P}_f(E) + \mathbb{P}_f(E')$, under two cases:
    \begin{enumerate}
        \item If $\mathbf{a} \notin \mathcal{C}^\perp$:
        It follows that $\mathcal{C} \not\subset \a^{\perp}$. In this case, exactly half of the elements in $\mathcal{C}$ fall within $\a^{\perp}$. The same argument holds for any coset, including $\mathbf{b}+\mathcal{C}$ and $\mathbf{b}'+\mathcal{C}$. Hence, from Eq.~\eqref{Eq:failureProbability}, we have
        \begin{equation*}
            \mathbb{P}_f(E) = \mathbb{P}_f(E') = \frac{2^{k-1}}{2^{n-1}} = 2^{k-n}.
        \end{equation*}
        Therefore,
        \begin{equation*}
            \mathbb{P}_f(E) + \mathbb{P}_f(E') = 2^{k-n+1}.
        \end{equation*}

        \item If $\mathbf{a} \in \mathcal{C}^\perp$: 
        The code $\mathcal{C}$ is contained within $\a^{\perp}$. It follows that for any codeword $\mathbf{c} \in \mathcal{C}$, the inner products $\mathbf{a} \cdot (\mathbf{b}+\mathbf{c}) = \mathbf{a} \cdot \mathbf{b}$, and $\mathbf{a} \cdot (\mathbf{b}'+\mathbf{c}) = \mathbf{a} \cdot \mathbf{b}' = \mathbf{a} \cdot \mathbf{b} + 1 \pmod 2$. Consequently, exactly one of the cosets, either $\mathbf{b}+\mathcal{C}$ or $\mathbf{b}'+\mathcal{C}$, is completely contained in $\a^{\perp}$, while the other is strictly disjoint from it. Therefore, one of the probabilities $\mathbb{P}_f(E)$ or $\mathbb{P}_f(E')$ equals $\frac{2^k}{2^{n-1}} = 2^{k-n+1}$, and the other equals $0$. Hence,
        \begin{equation*}
            \mathbb{P}_f(E) + \mathbb{P}_f(E') = 2^{k-n+1}.
        \end{equation*}
    \end{enumerate}

    In both cases, the sum of the failure probabilities for any pair $\{E, E'\}$ is exactly $2^{k-n+1}$. Substituting this into Eq.~\eqref{Eq:SucessProbUnionBound} yields:
    \begin{equation}\label{Eq:successProb}
        \begin{split}
            \mathbb{P} \left( \bigcap_{E \in \mathcal{E}} \{\operatorname{Cl}_{\mathrm{G}}(E) \notin \mathcal{C}\} \right) &\geq 1 - \frac{|\mathcal{E}|}{2} \left(2^{k-n+1}\right) \\
            &= 1 - 2^{k-n} \sum_{j=1}^{d-1} \binom{n}{j}(3^j-1).
        \end{split}
    \end{equation}
    Because $\operatorname{Cl}_{\mathrm{G}}(E) \notin \mathcal{C}$ for all $E \in \mathcal{E}$ is sufficient to guarantee that $\mathcal{Q}$ detects all errors of weight less than $d$, the probability that $\mathcal{Q}$ has minimum distance $d$ is bounded below by $1 - 2^{k-n} \sum_{j=1}^{d-1} \binom{n}{j}(3^j-1)$, concluding the proof.
\end{proof}

By requiring the lower bound on the success probability in Theorem~\ref{mainTheorem} to be strictly positive, we obtain the following corollary. 

\begin{corollary}\label{cor:CWS_existence}
    For every classical $[n,k,d]$ binary linear code $\mathcal{C}$ satisfying 
    \begin{equation}
        \sum_{j = 1}^{d-1} \binom{n}{j}(3^j-1) < 2^{n-k},
        \label{ineq:GV}
    \end{equation}
    there exists a graph $\mathrm{G}$ that yields an $[[n,k,d]]$ pure quantum stabilizer code $\mathcal{Q}$ via the CWS framework. 
\end{corollary}

An asymptotic analysis of this result yields a Gilbert-Varshamov-style bound, given by the following corollary. 

\begin{corollary}
    Let $\mathcal{C}$ be a classical binary linear code of blocklength $n$, rate $R = k/n$, and relative distance $\delta = d/n \le 3/4$, satisfying the condition $R + h_2(\delta) + \delta\log_2(3) < 1$. The probability that the random $[[n,nR]]$ quantum stabilizer code $\mathcal{Q} = \operatorname{CWS}(\mathcal{C}, \mathrm{G}(n, \frac12))$ is a pure code with distance $d = \delta n$ is bounded below by $1- 2^{-n[1-R-h_2(\delta)-\delta\log_2 3]}$. In particular, this probability approaches $1$ as $n \to \infty$.
    \label{cor:2}
\end{corollary}

\begin{proof}
    Let $d_{\mathcal{Q}}$ denote the minimum distance of the quantum code $\mathcal{Q}$. From Theorem~\ref{mainTheorem}, the probability that $d_{\mathcal{Q}}$ falls strictly below $d$ is bounded by:
    \begin{equation}\label{eq:asymptotic1}
        \begin{split}
            \mathbb{P}(d_\mathcal{Q} < d) &\leq 2^{k-n} \sum_{j=1}^{d-1} \binom{n}{j}(3^j - 1) \\
            &< 2^{k-n} \sum_{j=0}^{d} \binom{n}{j} 3^j.
        \end{split}
    \end{equation}
    We use the following standard upper bound on the volume of a $q$-ary Hamming sphere (see e.g., \cite[Lemma~4.7]{Roth2006})
    \begin{equation*}
        \sum_{j=0}^{t} \binom{n}{j} (q-1)^j \leq q^{n h_q(t/n)}, \quad \text{for }0 \le t/n \le 1-\frac{1}{q},
    \end{equation*}
    where $h_q(x) = x \log_q(q-1) - x \log_q x - (1-x) \log_q(1-x)$ is the $q$-ary entropy function. Applying this bound to Eq.~\eqref{eq:asymptotic1} with $q=4$ and $t=d = \delta n$, we obtain:
    \begin{equation*}
        \begin{split}
            \mathbb{P}(d_\mathcal{Q} < d) &< 2^{k-n} \cdot 4^{n h_4(\delta)} \\
            &= 2^{k-n} \cdot 2^{n \left(h_2(\delta) + \delta \log_2 3\right)} \\
            &= 2^{n \left(R + h_2(\delta) + \delta \log_2 3 - 1\right)},
        \end{split}
    \end{equation*}
    where we utilized the identity $2h_4(\delta) = h_2(\delta) + \delta\log_2 3$. 
%
\end{proof}

From the two corollaries above, we can deduce the following quantum GV bound.

\begin{corollary}
\mbox{ }
    \begin{itemize}
        \item[(a)] For any choice of parameters $n,k,d$ satisfying \eqref{ineq:GV}, there exists an $[[n,k,\ge d]]$ pure quantum stabilizer code.
        \item[(b)] For blocklength $n$, rate $R$ and relative distance $\delta\leq 3/4$ that satisfy $R+h_2(\delta)+\delta\log_2 3 < 1$, there exists\footnote{We assume, to keep the statement simple, that $nR$ and $n\delta$ are integers.} an $[[n, nR, \ge \delta n]]$ pure quantum stabilizer code.
    \end{itemize}
    \label{cor:GV}
\end{corollary}
\begin{proof}
    (a)~Note first that $3^j \ge 1+j \ge 1 + j/(n+1-j) = \binom{n+1}{j} / \binom{n}{j}$, and hence, $\binom{n}{j} 3^j \ge \binom{n+1}{j}$. With this, the left-hand-side of \eqref{ineq:GV} is lower bounded by
    $$
    \sum_{j=1}^{d-1} \left[\binom{n+1}{j} - \binom{n}{j}\right] = \sum_{j=1}^{d-1} \binom{n}{j-1} = \sum_{j=0}^{d-2}\binom{n}{j},
    $$
    which is in turn lower bounded by $\sum_{j=0}^{d-2}\binom{n-1}{j}$. So, if \eqref{ineq:GV} is satisfied, then so is the classical GV inequality $\sum_{j=0}^{d-2} \binom{n-1}{j} < 2^{n-k}$, which guarantees the existence of an $[n,k, \ge d]$ binary linear code. The existence of an $[[n,k,\ge d]]$ quantum stabilizer code now follows from Corollary~\ref{cor:CWS_existence}.

    (b)~If $n, R, \delta$ satisfy the condition in the statement, they also satisfy the weaker condition : $R < 1 - h_2(\delta)$. Therefore, by the classical asymptotic GV bound (see e.g., \cite[Theorem~4.10]{Roth2006}), an $[n,nR, \ge n\delta]$ binary linear code exists, and the statement follows from Corollary~\ref{cor:2}.
\end{proof}

Feng and Ma established the following GV bound for the existence of an $[[n,k,d]]$ pure quantum stabilizer code \cite{FengMa04}:
\begin{equation*}
    \sum_{j=1}^{d-1} \binom{n}{j} 3^j < 2^{n-k+2} - 1.
\end{equation*}
While their finite bound is less restrictive and guarantees existence for a broader set of parameters, our bound provides the additional structural guarantee that these pure quantum codes can be constructed directly from classical binary linear codes with the same parameters via the CWS framework. Furthermore, the asymptotic version of our bound is identical to their asymptotic bound for pure quantum codes, which also agrees with the asymptotic bound proposed for general quantum stabilizer codes by Ekert and Macchiavello \cite{ekert}.

\section{Conclusion}\label{sec:conclusion}
We investigated the construction of quantum stabilizer codes within the codeword stabilized (CWS) code framework, using classical binary linear codes and random graphs. Specifically, we addressed whether an arbitrary classical $[n,k,d]$ binary linear code can be lifted to an $[[n,k,d]]$ quantum code by choosing a suitable graph. By sampling the graph from an Erd\H{o}s-R\'enyi random graph model, we established a lower bound on the probability that the resulting CWS code preserves the minimum distance $d$ of the classical code. This probabilistic analysis yielded a GV-style bound on the code parameters that guarantees the existence of a graph $\mathrm{G}$ achieving this minimum distance.

Furthermore, we extended our analysis to the asymptotic regime. We showed that for any family of classical codes with rate $R$ and relative distance $\delta$ satisfying $R + h_2(\delta) + \delta\log_2 3 < 1$, random CWS codes constructed on Erd\H{o}s-R\'enyi random graphs form, with high probability, a family of quantum codes of rate $R$ and relative distance $\delta$.

A natural open question following this existence proof is how to explicitly and efficiently construct these graphs. As a starting point, one could investigate specific families of good classical codes to identify compatible, deterministic graph structures. We leave this as a direction for future work.

\section*{Statement of AI Usage}
During the preparation of this manuscript, the authors used Gemini 3.1 Pro and GPT-5.6 Luna to refine the presentation of the proofs and polish the writing. The core technical content, analysis, and conclusions were entirely generated and reviewed by the human authors, who take full responsibility for the content and all claims made in this manuscript.

\section*{Acknowledgment}
This work was supported by the Prime Minister's Research Fellowship (PMRF) awarded to T. Aswanth, and by the ANRF MATRICS grant ANRF/ARGM/2025/000814/MTR awarded to N.\ Kashyap.


%
\bibliographystyle{IEEEtran}
\bibliography{bibliofile}

%








\end{document}